\documentclass[letterhead,11pt]{article}

\usepackage[utf8]{inputenc}
\usepackage[T1]{fontenc}
\usepackage{graphicx}
\usepackage{amsmath}
\usepackage{float}
\usepackage{amsfonts}
\usepackage[dvipsnames,usenames]{color}
\usepackage[colorlinks=true,urlcolor=Blue,citecolor=Green,linkcolor=BrickRed]{hyperref}
\usepackage[usenames,dvipsnames]{xcolor}
\usepackage[noadjust]{cite}
\usepackage{amsmath}
\usepackage{fullpage}
\usepackage{amsmath,amsthm}
\usepackage{todonotes}
\usepackage{thmtools,thm-restate}

\newcommand{\Oh}{O}
\newcommand{\weight}{{\rm \omega}}
\newcommand{\VD}{{\rm VD}}
\newcommand{\Vor}{{\rm Vor}}
\newcommand{\dist}{d}
\newtheorem{lemma}{Lemma}
\newtheorem{theorem}[lemma]{Theorem}
\newtheorem{proposition}[lemma]{Proposition}
\newtheorem{corollary}[lemma]{Corollary}

\begin{document}

\title{Better Tradeoffs for Exact Distance Oracles in Planar Graphs\thanks{This research was supported in part by Israel Science Foundation grant 794/13.}
}
\author{Paweł Gawrychowski, Shay Mozes, Oren Weimann, and Christian Wulff-Nilsen}
\date{}

\maketitle

\begin{abstract}
We present an $O(n^{1.5})$-space distance oracle for directed planar graphs that answers distance queries in $O(\log n)$ time. 
Our oracle both significantly simplifies and significantly improves the recent oracle of Cohen-Addad, Dahlgaard and Wulff-Nilsen [FOCS 2017], which uses $O(n^{5/3})$-space and answers queries in $O(\log n)$ time. We achieve this by designing an elegant and efficient point location data structure for Voronoi diagrams on planar graphs. 

We further show a smooth tradeoff between space and query-time. 
For any $S\in [n,n^2]$, we show an oracle of size $S$ that answers queries in $\tilde O(\max\{1,n^{1.5}/S\})$ time. 
This new tradeoff is currently the best (up to polylogarithmic factors) for the entire range of $S$
and improves by polynomial factors over all the previously known tradeoffs for the range $S \in [n,n^{5/3}]$.

\end{abstract}

\thispagestyle{empty}
\clearpage
\setcounter{page}{1}

\section{Introduction}
Computing shortest paths is a classical and fundamental algorithmic problem that has received considerable attention from the research community for decades. A natural data structure problem in this context is to compactly store information about distances in a graph in such a way that the distance between any pair of query vertices can be computed efficiently. A data structure that support such queries is called a {\em distance oracle}. Naturally, there is a tradeoff between the amount of space consumed by a distance oracle, and the time required by distance queries. Another quantity of interest is the preprocessing time required for constructing the oracle.

\paragraph{Distance oracles in planar graphs.}
It is natural to consider the class of planar graphs in this setting since planar graphs arise in many important applications involving distances, most notably in navigation applications on road maps. Moreover, planar graphs exhibit many structural properties that facilitate the design of very efficient algorithms. Indeed, distance oracles for planar graphs have been extensively studied. These oracles can be divided into two groups: {\em exact} distance oracles which always output the correct distance, and {\em approximate} distance oracles which allow a small stretch in the distance output. 
For approximate distance oracles, one can obtain near-linear space and near-constant query-time at the cost of a $(1+\epsilon)$ stretch (for any fixed $\epsilon$)~\cite{Thorup04,Klein2002,KawarabayashiST13,SommerICALP11,Wulff-Nilsen16}. 
In this paper we focus on the tradeoff between space and query-time of exact distance oracles for planar graphs.

\paragraph{Exact distance oracles.} The following results as well as ours all hold for directed planar graphs with real arc-lengths (but no negative length cycles). Djidjev~\cite{Djidjev96} and Arikati et al.~\cite{Arikati96} obtained distance oracles with the following tradeoff between space and query-time. For any $S\in [n,n^2]$, they show an oracle with space $S$ and query-time of $O(n^2/S^2)$. For $S\in [n^{4/3}, n^{1.5}]$, Djidjev's oracle achieves an improved bound of $O(n/\sqrt S)$. This bound (up to polylogarithmic factors) was extended to the entire range $S\in [n,n^2]$ in a number of papers~\cite{Xu2000,Cabello06,Nussbaum11,FR06,MozesSommer}. Wulff-Nilsen~\cite{WNthesis} showed how to achieve constant query-time with $O(n^2(\log\log n)^4/\log n)$ space, improving the above tradeoff for close to quadratic space. Very recently, Cohen-Addad, Dahlgaard, and Wulff-Nilsen~\cite{Cohen-AddadDW17}, inspired by the ideas of Cabello~\cite{Cabello17} made significant progress by presenting an oracle with $O(n^{5/3})$ space and $O(\log n)$ query-time. This is the first oracle for planar graphs that achieves truly subquadratic space and subpolynomial query-time. They also showed that with $S \ge n^{1.5}$ space, a query-time of $O(n^{2.5}/S^{1.5}\log n)$ is possible.
To summarize, prior to the results described in the current paper the best known tradeoff was $\tilde O(n/\sqrt S)$ query-time for space $S \in [n,n^{1.5}]$, $\tilde O(n^{2.5}/S^{1.5})$ query-time for space $S \in [n^{1.5},n^{5/3}]$, and $O(\log n)$ query-time for $S \in [n^{5/3},n^2]$.

\paragraph{Our results and techniques.}
\label{subsec:VorDiagramOracle}

In this paper we 
show a distance oracle with $O(n^{1.5})$ space and $O(\log n)$ query-time.
More generally, for any $r\leq n$ we construct a distance oracle with $O(n^{1.5}/\sqrt{r}+n\log r\log(n/r))$ space
and $O(\sqrt{r}\log n\log r)$ query-time.
This improves the currently best known tradeoffs for essentially the entire range of $S$: for space $S \in [n,n^{1.5}]$ we obtain an oracle with $\tilde O(n^{1.5}/S)$ query-time, while for space $S \in [n^{1.5},n^2]$, our oracle has query-time of $O(\log n)$. 

To explain our techniques we need the notion of an additively weighted Voronoi diagram on a planar graph.
Let $P=(V,E)$ be a directed planar graph, and let $S\subseteq V$ be a subset of the vertices, which are called the {\em sites} of the Voronoi diagram.
Each site $u\in S$ has a weight $\weight(u)\geq 0$ associated with it.
The distance between a site $u\in S$ and a vertex $v\in V$, denoted
by $\dist(u,v)$, is defined as $\weight(u)$ plus the length of the $u$-to-$v$ shortest path in $P$. 
The {\em additively weighted Voronoi diagram} of $(S,\weight)$ within $P$,
denoted $\VD(S,\weight)$, is a partition of $V$ into pairwise disjoint
sets, one set $\Vor(u)$ for each site $u\in S$. The set $\Vor(u)$,
which is called the {\em Voronoi cell} of $u$, contains all vertices in $V$ that are closer (w.r.t.~$\dist(\cdot,\cdot)$) to $u$ than to any other site in $S$
(assuming that the distances are unique). 
There is a dual representation $\VD^*(S,\weight)$ of a Voronoi diagram $\VD(S,\weight)$ as a planar graph with $O(|S|)$ vertices and edges. See Section~\ref{sec:prelims}.

We obtain our results using point location in additively weighted Voronoi diagrams. This approach is also the one taken in~\cite{Cohen-AddadDW17}. However, our construction is arguably simpler and more elegant than that of~\cite{Cohen-AddadDW17}. 
Our main technical contribution is a novel point location data structure for Voronoi diagrams (see below). 
Given this data structure, the description of the preprocessing and query algorithms of our $O(n^{1.5})$-space oracle are extremely simple and require a few lines each. 
In a nutshell, the construction is recursive, using simple cycle separators. We store a Voronoi diagram for each node $u$ of the graph. The sites of this diagram are the vertices of the separator and the weights are the distances from $u$ to each site. To get the distance from $u$ to $v$ it suffices to locate the node $v$ in the Voronoi diagram stored for $u$ using the point location data structure. Since the cycle separator has $O(\sqrt n)$ vertices, this yields an oracle requiring $O(n^{1.5})$  space.  

The oracles for the tradeoff are built upon this simple oracle by storing Voronoi diagrams for just a subset of the nodes in a graph (the so called boundary vertices of an $r$-division). This requires less space, but the query-time increases. This is because a node $u$ now typically does not have a dedicated Voronoi diagram. Therefore, to find the distance from $u$ to $v$, we now we need to locate $v$ in multiple Voronoi diagrams stored for nodes in the vicinity of $u$.

As we mentioned above, our main technical tool 
is a data structure that supports   
point location queries in Voronoi diagrams in $O(\log n)$ time. This is summarized in the following theorem.

\begin{theorem}\label{thm:extended}
Let $P$ be a directed planar graph with real arc-lengths, $r$ vertices, and no negative length cycles. Let $S$ be a set of $b$ sites that lie on a single face (hole) of $P$.
We can preprocess $P$ in $\tilde O(b\cdot r)$ time and $O(b\cdot r)$ space so that, given the dual representation of any additively weighted Voronoi diagram, $VD^*(S,\weight)$, we can extend it in $O(b)$ time and space to support the following queries. 
Given a vertex $v$ of $P$, report in $O(\log b)$ time the site $u$ such that $v$ belongs to $\Vor(u)$.

\end{theorem}
A data structure for the same task was described in~\cite{Cohen-AddadDW17}. Our data structure is both significantly simpler and more efficient.
Roughly speaking, the idea is as follows. We prove that $\VD^*(S,\weight)$ is a ternary tree. This allows us to use a straightforward centroid decomposition of depth $O(\log b)$ for point location. To locate the voronoi cell $\Vor(u)$ containing a node $v$ we traverse the centroid decomposition. At any given level of the decomposition we only need to known which of the three subtrees in the next level contains $\Vor(u)$. To this end we associate with each centroid node three shortest paths. These paths partition the plane into three parts, each containing exactly one of the three subtrees. Identifying the desired subtree then boils down to determining the position of $v$ relative to these three shortest paths. We show that this can be easily done by examining the preorder number of $v$ in the shortest path trees rooted at three sites.  

\paragraph{Roadmap.}  Theorem~\ref{thm:extended} is proved in Section~\ref{sec:pointlocation}  under a simplifying assumption
that every site belongs to its Voronoi cell. This  suffices to design our distance oracle with
space $O(n^{1.5})$ and query-time $O(\log n)$, which is done in Section~\ref{sec:overall}, assuming that Theorem~\ref{thm:extended} holds. Section~\ref{sec:tradeoff} describes
the improved space to query-time tradeoff. Finally, in Section~\ref{sec:extension} we describe how to remove the simplifying assumption.
Additional details and some omitted proofs appear in the appendix. 

\section{Preliminaries}\label{sec:prelims}
We assume that shortest paths are unique. 
This can be ensured in linear time by a random perturbation of the edge lengths~\cite{Isolation,Isolation2} or deterministically in near-linear time using lexicographic comparisons~\cite{CabelloCE13,HartvigsenMardon}.
It will also be convenient to assume that graphs are strongly connected; if not, we can always triangulate them with bidirected edges of infinite length.

\paragraph{Separators in planar graphs.} Given a planar embedded graph $G$, a Jordan curve separator is a simple closed curve in the plane that intersects the embedding of $G$ only at vertices. 
Miller~\cite{Miller86} showed that any $n$-vertex 
planar embedded graph has a Jordan curve separator of size $O(\sqrt n)$ such that the number of vertices on each side of the curve is at most $2n/3$. In fact, the balance of $2/3$ can be achieved with respect to any weight function on the vertices, not necessarily the uniform one. Miller also showed that the vertices of the separator ordered along the curve can be computed in $O(n)$ time.
An $r$-{\em division}~\cite{F87} of a planar graph $G$, for some  $r \in (1,n)$, is a decomposition of $G$
 into $O(n/r)$ pieces,  
where each piece has at most $r$ vertices and $O(\sqrt{r})$ \emph{boundary} vertices (vertices shared with other pieces).
There is an $O(n)$ time algorithm that computes an $r$-division 
of a planar graph with the additional property
that, in every piece, the number of faces of the piece that are not faces of the original graph $G$ is constant~\cite{KMS13,AvWZ13} (such faces are called {\em holes}). 

\paragraph{Voronoi diagrams on planar graphs.}
Recall the definition of additively weighted Voronoi diagrams $\VD(S,\weight)$ from the introduction. We write just $\VD$ when the particular $S$ and $\weight$ are not important, or when they are clear from the context.

We restrict our discussion to the case where the sites $S$ lie on a single face, denoted by $h$. 
We work with a dual representation of $\VD(S,\weight)$, denoted $\VD^{*}(S,\weight)$ or simply $\VD^{*}$. 
Let $P^*$ be the planar dual of $P$. Let $\VD^*_0$ be the subgraph of $P^*$ consisting of the duals of edges $uv$ of $P$ such that $u$ and $v$ are in different Voronoi cells. Let $\VD_1^*$ be the graph obtained from $\VD^*_0$ by contracting edges incident to a degree-2 vertex one after the other until no degree 2 vertices remain. 
The vertices of $\VD^*_1$ are called Voronoi vertices. A Voronoi vertex $f^*$ is dual to
a face $f$ such that the nodes incident to $f$ belong to at least three
different Voronoi cells. 
In particular, $h^{*}$ (i.e., the dual vertex corresponding to the face $h$ to which all the sites are incident) is a Voronoi vertex.
Each face of $\VD_1^{*}$ corresponds to a cell
$\Vor(v_{i})$, hence there are at most $|S|$ Voronoi vertices, and, by sparsity of planar graphs, the complexity (i.e., the number of nodes, edges and faces) of $\VD_1^*$ is $O(|S|)$. 
Finally, we define $\VD^*$ to be the graph obtained from $\VD_1^{*}$ after replacing the node $h^{*}$ by multiple
copies, one for each incident edge. The original Voronoi vertices are called
{\em real}. See Figure~\ref{fig:tree}. 

Given a planar graph $P$ with $r$ nodes and a set $S$ of $b$ sites on a single face $h$, one can compute any additively weighted Voronoi diagram $\VD(S,\weight)$ naively in $\tilde O(r)$ time by adding an artificial source node, connecting it to every site $s$ with an edge of length $\weight(s)$, and computing the shortest path tree. The  
dual representation $\VD^*(S,\weight)$ can then be obtained in additional $O(r)$ time by following the constructive description above. There are more efficient algorithms~\cite{Cabello17,GKMSW17} when one wants to construct many different additively weighted Voronoi diagrams for the same set of sites $S$. The basic approach is to invest superlinear time in preprocessing $P$, but then construct $\VD(S,\weight)$ for multiple choices of $\weight$ in $\tilde O(|S|)$ each instead of $\tilde O(r)$. Since the focus of this paper is on the tradeoff between space and query-time, and not on the preprocessing time, the particular algorithm used for constructing the Voronoi diagrams is less important. 

\begin{figure}[h]
\begin{center}
\includegraphics[width=0.65\textwidth]{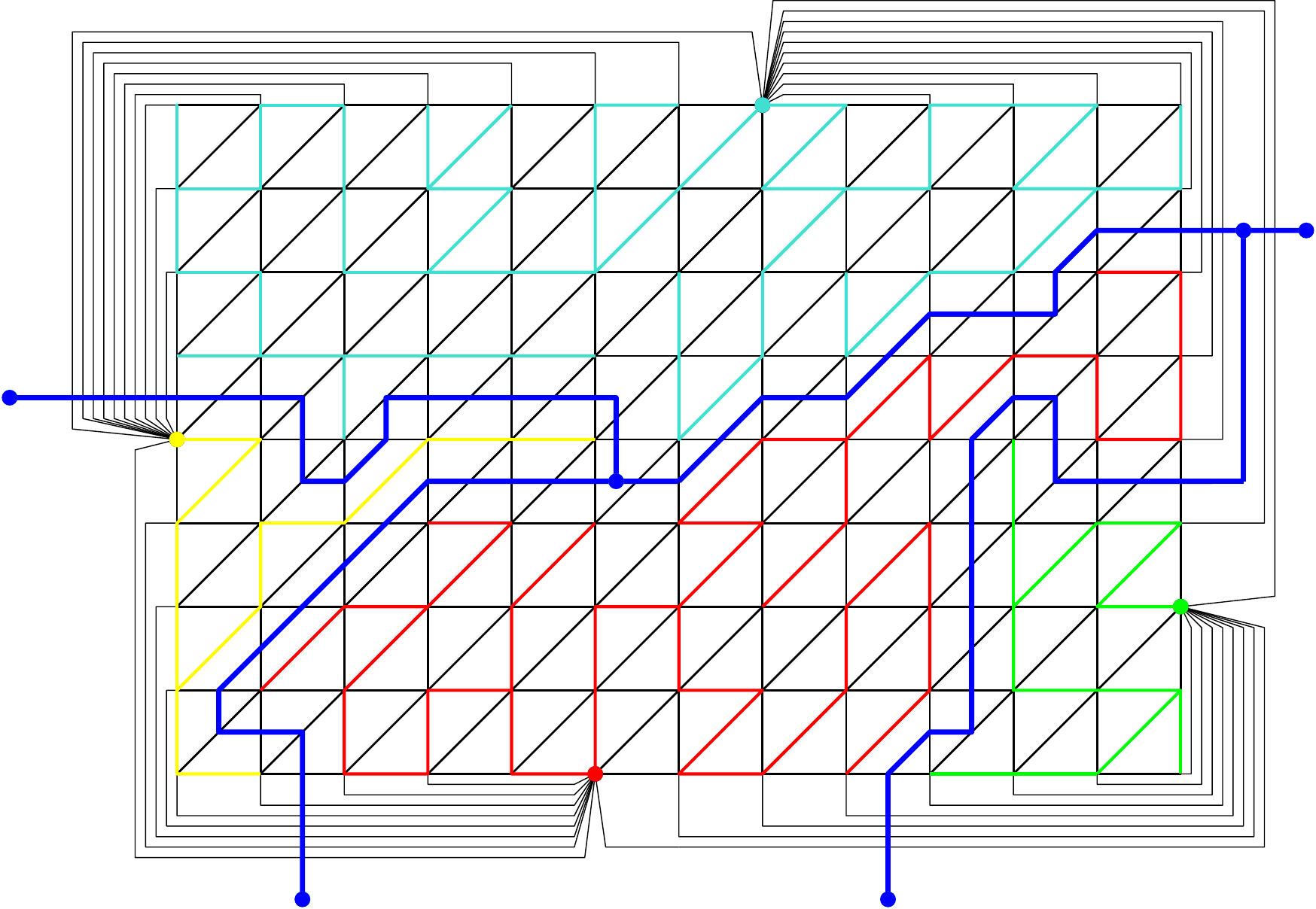}
\end{center}
\caption{A planar graph (black edges) with four sites on the infinite face together with  the dual Voronoi diagram $\VD^{*}$ (in blue). The sites are shown together with their corresponding shortest path trees (in turquoise, red, yellow, and green).
Two of the Voronoi vertices (in blue) are real.\label{fig:tree}}
\end{figure}

\section{The Oracle}\label{sec:overall}
In this section we describe our distance oracle assuming Theorem~\ref{thm:extended}. 
Let $G$ be a directed planar graph with non-negative arc-lengths. 
At a high level, our oracle is based on a recursive decomposition
of $G$ into pieces using Jordan curve separators.
Each piece $P=(V,E)$ is a subgraph of $G$. The boundary vertices of $P$ are vertices of $P$ that are incident (in $G$) to edges not in $P$. The holes of $P$ are faces of $P$ that are not faces of $G$. Note that every boundary vertex of $P$ is incident to some hole of $P$.

A piece $R=(V,E)$ is decomposed into two smaller pieces on the next
level of the decomposition as follows. We choose a Jordan curve separator
$C=(v_{1},v_{2},\ldots,v_{k})$, where $k=O(\sqrt{|V|})$. This separates
the plane into two parts and defines two smaller pieces $P$ and $Q$
corresponding to, respectively, the subgraphs of $R$ inside and the outside of $C$. Every edge of
$R$ is assigned to either $P$ or $Q$.
Thus, on every level of the recursive decomposition into pieces, an edge
of $G$ appears in exactly one piece.
The separators in levels congruent to 0 modulo 3 are chosen to balance the total number of nodes. The separators in levels congruent to 1 modulo 3 are chosen to balance the number of boundary nodes. The separators in levels congruent to 2 modulo 3 are chosen to balance the number of holes. This guarantees  that the number of holes in each piece is constant, and that the number of vertices and boundary vertices decrease exponentially along the recursion. In particular, the depth of the decomposition is logarithmic in $|V|$.
These properties are summarized in the following lemma whose proof is in the appendix.

\begin{restatable}{lemma}{separating}
\label{lem:separating}
Choosing the separators as described above guarantees that (i) each piece
has $O(1)$ holes, (ii) the number of nodes in a piece on the $\ell$-th
level in the decomposition is $O(n/c_{1}^{\ell/3})$, for some constant
$c_{1}>1$, (iii) the number of boundary nodes in a piece on the $\ell$-th
level in the decomposition is $O(\sqrt{n}/c_{2}^{\ell/3})$, for some constant $c_{2}>1$.
\end{restatable}

\paragraph{Preprocessing.}
We compute a recursive decomposition of $G$ using Jordan separators as described above. 
For each piece $R=(V_R,E_R)$ in the recursive decomposition we perform the following preprocessing. 
We compute and store, for each boundary node $v$ of $R$, the shortest path tree $T^R_v$ in $R$ rooted at $v$.
Additionally, we store for every node $u$ of $R$ the distance from $v$ to $u$  and the distance from $u$ to $v$ in the whole $G$.
For a non-terminal piece $R$, let $P=(V_{P},E_{P})$ and $Q=(V_{Q},E_{Q})$ be the two pieces into which $R$ is separated.
For every node $u\in V_{Q}$ and for every hole $h$ of $P$ we store
an additively weighted Voronoi diagram $\VD(S_h,\weight)$ for $P$, where the set of sites $S_h$ is the set of boundary nodes of $P$ incident to the hole $h$, and the additive weights $\weight$ correspond to the distances in $G$ from $u$ to each site in $S_h$. We enhance each Voronoi diagram with the point location data structure of Theorem~\ref{thm:extended}.
We also store the same information with the roles of $Q$ and $P$ exchanged.

\paragraph{Query.}
To compute the distance from $u$
to $v$, we traverse the recursive decomposition starting from the piece that  corresponds to the
whole initial graph $G$. Suppose that the current piece is $R=(V,E)$,
which is partitioned into $P$ and $Q$ with a Jordan curve separator
$C$. If $v\in C$ then, because the nodes of $C$ are boundary nodes
in both $P$ and $Q$, we return the additive weight $\weight(v)$ in the Voronoi diagram stored for $u$, which is equal
to the distance from $u$ to $v$ in $G$. Similarly, if $u\in C$ then we retrieve
and return the distance from $u$ to $v$ in the whole $G$.
The remaining case is that both $u$ and $v$ belong
to a unique piece $P$ or $Q$. If both $u$ and $v$ belong
to the same piece on the lower level of the decomposition, we continue
to that piece. Otherwise, assume without loss of generality that $u\in Q$
and $v\in P$. Then, the shortest path from $u$ to $v$ must
go through a boundary node $v_{i}$ of $P$. We therefore perform a point location query for $v$ in each of the Voronoi diagrams stored for $u$ and for some hole $h$ of $P$. Let $s_1,\dots,s_g$ be the sites returned by these queries, where $g=O(1)$ is the number of holes of $P$. The distance in $G$ from $u$ to $s_i$ is $\weight(s_i)$, and the distance in $P$ from $s_i$ to $v$ is stored in $T^P_{s_i}$. We compute the sum of these two terms for each $s_i$, and return the minimum sum computed.

\paragraph{Analysis.}
First note that the query-time is $O(\log n)$ since, at each step of the traversal, we either descend to a smaller piece in $O(1)$ time or terminate after having found the desired distance in $O(\log n)$ time by $O(1)$ queries to a point location structure.

Next, we analyze the space.
Consider a piece $R$
with $O(1)$ holes.
Let $n(R)$ and $b(R)$ denote the number of nodes and boundary nodes of $R$, respectively.
The trees $T^R_u$ and the stored distances in $G$ require a total of $O(b(R)\cdot n(R))$ space.
Let $R$ be further decomposed into pieces $P$ and $Q$.
We bound the space used by all Voronoi diagrams created for $R$.
Recall that every Voronoi diagram and point location
structure corresponds to a node $u$ of $P$ and a hole of $Q$, or
vice versa. 
The size of each additively weighted Voronoi diagram stored for a node of $P$
is $O(b(Q))$, so $O(n(P)\cdot b(Q))$ for all nodes of $P$.
The additional space required by Theorem~\ref{thm:extended} is 
also $O(n(P)\cdot b(Q))$.
Finally, for every node of $R$ we record if it belongs to the Jordan curve
separator used to further divide $R$, and, if not, to which of the resulting two pieces it belongs. This takes only $O(n(R))$ space.
The total space for each piece $R$ is thus $O(n(R)\cdot b(R))$ plus
$O(n(P)\cdot b(Q)+n(Q)\cdot b(P))$ if $R$ is decomposed into $P$ and $Q$.

We need to bound the sum of $O(n(R)\cdot b(R))$ over all the pieces $R$.
Consider all pieces $R_{1},R_{2},\ldots,R_{s}$ on the same level $\ell$ in the decomposition.
Because these pieces are edge-disjoint, $\sum_{i}n(R_{i})=O(n)$.
Additionally, $b(R_{i})=O(\sqrt{n}/c^{\ell})$ for any $i$, where $c>1$,
so $\sum_{i} O(n(R_{i})\cdot b(R_{i})) = O(n^{1.5}/c^{\ell})$.
Summing over all levels $\ell$, this is $O(n^{1.5})$. The sum
of $O(n(P)\cdot b(Q)+n(Q)\cdot b(P))$ over all pieces $R$ that are decomposed
into $P$ and $Q$ can be analysed with the same reasoning
to obtain that the total size of the oracle is $O(n^{1.5})$.

Finally, we analyze the preprocessing time. For each piece $R$, the preprocessing of Theorem~\ref{thm:extended} takes $\tilde O(n(R)\cdot b(R))$.
Then, we compute $O(n(R))$ different additively weighted Voronoi diagrams for $R$. Each diagram is built in $\tilde O(n(R))$ time, and
its representation is extended in $O(b(R))$ time to support point location with Theorem~\ref{thm:extended}. The total preprocessing time
for $R$ is hence $\tilde O((n(R))^2)$, which sums up to $\tilde O(n^{2})$ overall by Lemma~\ref{lem:separating}.
We also need to compute the distances between pairs of vertices in $G$. This can be also done in $\tilde O(n^{2})$ 
total time by computing the shortest path tree rooted at each vertex in $\tilde O(n)$~\cite{KleinMW10}.

\section{Point Location in Voronoi Diagrams}\label{sec:pointlocation}

In this section we prove Theorem~\ref{thm:extended}. 
Let $P$ be a piece (i.e., a planar graph), and $S$ be a set of sites that lie on a single face (hole) $h$ of $P$. 

Our goal is to preprocess $P$ once in $O(|P||S|)$ time and space, and then, given any additively weighted Voronoi diagram
$\VD^*(S,\weight)$ (denoted $\VD^*$ for short), preprocess it in $O(|S|)$ time and space so as to answer point location queries in $O(\log |S|)$ time. 

We assume that the hole $h$ incident to all nodes in $S$ is the external face.
We assume that
all nodes of $P^*$, except for $h^{*}$, have degree 3. 
This can be achieved by triangulating $P$ with infinite length edges.
We also assume that all nodes
incident to the external face belong to $S$ and denote them
$s_{1},s_{2},\ldots,s_{b}$, according to their clockwise order on $h$. We assume $b\geq 3$ (for any constant $b$ point location is trivial).

Recall that for a site $u$ and a vertex $v$ we define $\dist(u,v)$ as $\weight(u)$ plus the length of the $u$-to-$v$ shortest path in $P$.
We further assume that no Voronoi cell is empty. That is, we assume that, for every pair of distinct sites $u,u' \in S$, $\weight(u) < \dist(u',u)$. 
If this assumption does not hold, let $S'$ be the subset of the sites whose Voronoi cells are non empty. We can embed inside the hole $h$ infinite length edges between every pair of consecutive sites in $S'$, and then again triangulate with infinite length edges.
This results in a new face $h'$ whose vertices are the sites in $S'$. Replacing $S$ with $S'$ and $h$ with $h'$ enforces the assumption. 
Note that since this transformation changes $P$, it is not suitable when working with Voronoi diagrams constructed by algorithms that preprocess $P$, such as the ones in~\cite{Cabello17,GKMSW17} 
In Section~\ref{sec:extension} we prove Theorem~\ref{thm:extended} without this assumption.

\subsection{Preprocessing for $P$}
The preprocessing for $P$ consists of computing shortest path trees $T_v$ for every boundary node $v \in S$, decorated with some additional information which we describe next. We stress that the additional information does not depend on any  weights $\weight$ (which are not available at preprocessing time).

Let $T_i$ be the shortest path tree in $P$ rooted at $s_i$. 
For a technical reason that will become clear soon, we add some artificial vertices to $T_i$. For each face $f$ of $P$ other than $h$, we add an artificial vertex $v_f$ whose embedding coincides with the embedding of the dual vertex $f^*$. Let $y_f$ be closest vertex to $s_i$ in $P$ that is incident to $f$. We add a zero length arc $y_fv_f$ to $T_i$. Note that $v_f$ is a leaf of $T_i$. Let $p_{i,f}$ be the shortest $s_i$-to-$v_f$ path in $T_i$. We say that a vertex $v$ of $T_i$ is to the right (left) of $p_{i,f}$ if the shortest $s_i$-to-$v$ path emanates right (left) of $p_{i,f}$. Note that, since $v_f$ is a leaf of $T_i$, $v$ is either right of $p_{i,f}$, left of $p_{i,f}$, or a vertex of $p_{i,f}$; the goal of adding the artificial vertices $v_f$ is to guarantee that these are the only options. The following proposition can be easily obtained using preorder numbers and a lowest common ancestor (LCA) data structure~\cite{BenderLCA} for $T_i$.

\begin{proposition}\label{lem:preorder}
There is a data structure with $O(|P|)$ preprocessing time that can decide in $O(1)$ time if for a given query vertex $v$ and query face $f$, $v$ is right of $p_{i,f}$, left of $p_{i,f}$, or a vertex of $p_{i,f}$.
\end{proposition}

We compute and store the shortest path trees $T_i$ rooted at each site $s_i$, along with preorder numbers and LCA data structures required by Proposition~\ref{lem:preorder}. This requires preprocessing time $\tilde O(|P||S|)$ by computing each $T_{i}$ in $\tilde O(|P|)$ time~\cite{KleinMW10},
and can be stored in $O(|P||S|)$ space. 

\subsection{Handling a Voronoi diagram $\VD^*(S,\weight)$}

We now describe how to handle an additively weighted  Voronoi diagram $\VD^* = \VD^*(S,\weight)$. This consists of a preprocessing stage and a query algorithm.
Handling $\VD^*$ crucially relies on the fact that, under assumption that each site is in its own Voronoi cell, $\VD^*$ is a tree.  
\begin{lemma}
\label{lem:tree}
$\VD^{*}$ is a tree.
\end{lemma}

\begin{proof}
Suppose that $\VD^*$ contains a cycle $C^*$. Since the degree of each copy of $h^*$ is one, the cycle does not contain $h^*$. Therefore, since all the sites are on the boundary of the hole $h$, the vertices of $P$ enclosed by $C^*$ are in a Voronoi cell that contains no site, a contradiction.

To prove that $\VD^*$ is connected, observe that in $\VD_1^*$, every Voronoi cell is a face (cycle) going through $h^*$. Let $C^*$ denote this cycle. If $C^*$ is disconnected in $\VD^*$ then, in $\VD_1^*$, $C^*$ must visit $h^*$ at least twice. But this implies that the cell corresponding to $C^*$ contains more than a single site, contradiction our assumption. Thus, the boundary of every Voronoi cell is a connected subgraph of $\VD^*$. Since the boundaries of the cell of $s_i$ and the cell of $s_{i+1}$ both contain the dual of the edge $s_is_{i+1}$, it follows that the entire modified $\VD^*$ is connected.
\end{proof}

We briefly describe the intuition behind the design of the point location data structure. To find the Voronoi cell $\Vor(s)$ to which a query vertex $v$ belongs, it suffices to identify an edge $e^*$ of $\VD^*$ that is adjacent to $\Vor(s)$. Given $e^*$ we can simply check which of its two adjacent cells contains $v$ by comparing the distances from the corresponding two sites to $v$.
Our point location structure is based on a {\em centroid decomposition} of $\VD^*$ into connected subtrees, and on the ability to determine, in constant time, which of the subtrees is the one that contains the desired edge $e^*$.

\paragraph{Preprocessing.}
The preprocessing consists of just computing a centroid decomposition of $\VD^*$. 
A {\em centroid} of an $n$-node tree $T$ is a a node $u\in T$ such that removing $u$ and replacing it with copies, one for each edge incident to $u$, results in a set of trees, each with at most $\frac{n+1}{2}$ edges. A centroid always exists in a tree with more than one edge.
In every step of the centroid decomposition of $\VD^*$, we work with
a connected subtree $T^*$ of $\VD^*$. Recall that 
there are no nodes of degree 2 in $\VD^*$.
If there are
no nodes of degree 3, then $T^*$ consists of a single edge of $\VD^*$, and the decomposition terminates. Otherwise, we choose a centroid $f^*$, and partition $T^*$
into the three subtrees $T^*_{0},T^*_{1},T^*_{2}$
obtained by splitting $f^*$ into three copies, one for each edge incident to $f^*$. Clearly, the depth of the recursive decomposition is $O(\log |S|)$. The decomposition can computed in $O(|S|)$ time and be represented as a ternary tree, which we call the {\em decomposition tree}, in $O(|S|)$ space.

\paragraph{Point location query.}
We first describe the structure that gives rise to the efficient query, and only then describe the query algorithm.
Consider a centroid $f^*$ used at some step of the decomposition. Let $s_{i_0}, s_{i_1}, s_{i_2}$ denote the three sites adjacent to $f^*$, listed in clockwise order along $h$. Let $f$ be the face of $P$ whose dual is $f^*$. Let $y_0, y_1, y_2$ be the three vertices of $f$, such that $y_j$ is the vertex of $f$ in $\Vor(s_{i_j})$.  Let $e^*_j$ be the edge of $VD^*$ incident to $f^*$ that is on the boundary of the Voronoi cells of $s_{i_j}$ and $s_{i_{j-1}}$ (indices are modulo 3). Let $T^*_j$ be the subtree of $T$ that contains $e^*_j$. Let $p_j$ denote the shortest $s_j$-to-$v_f$ path. Note that the vertex preceding $v_f$ in $p_j$ is $y_j$. See Figure~\ref{fig:schematic} (right).

\begin{figure}[h]
\begin{center}
\includegraphics[width=0.57\textwidth]{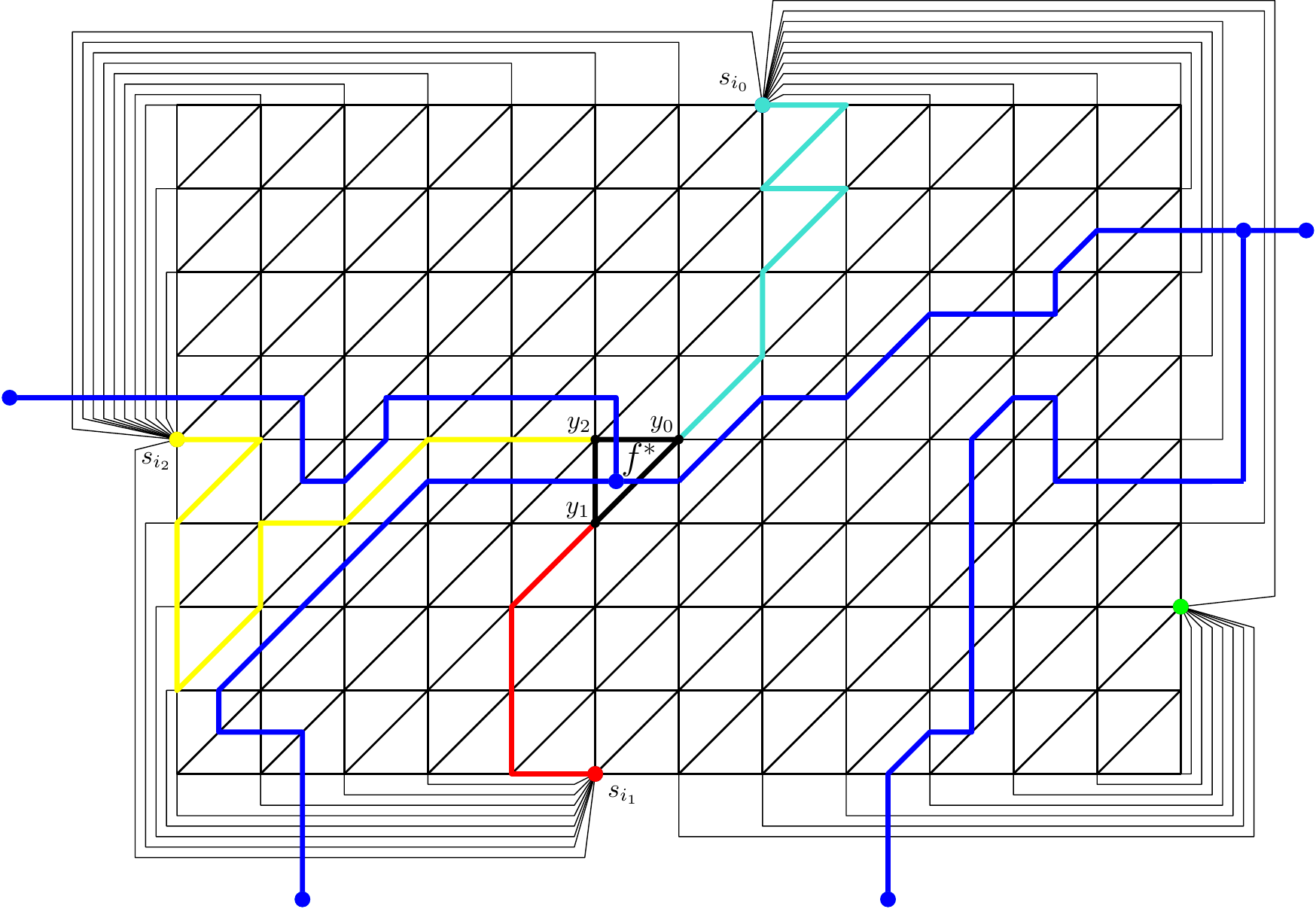}
\includegraphics[width=0.42\textwidth]{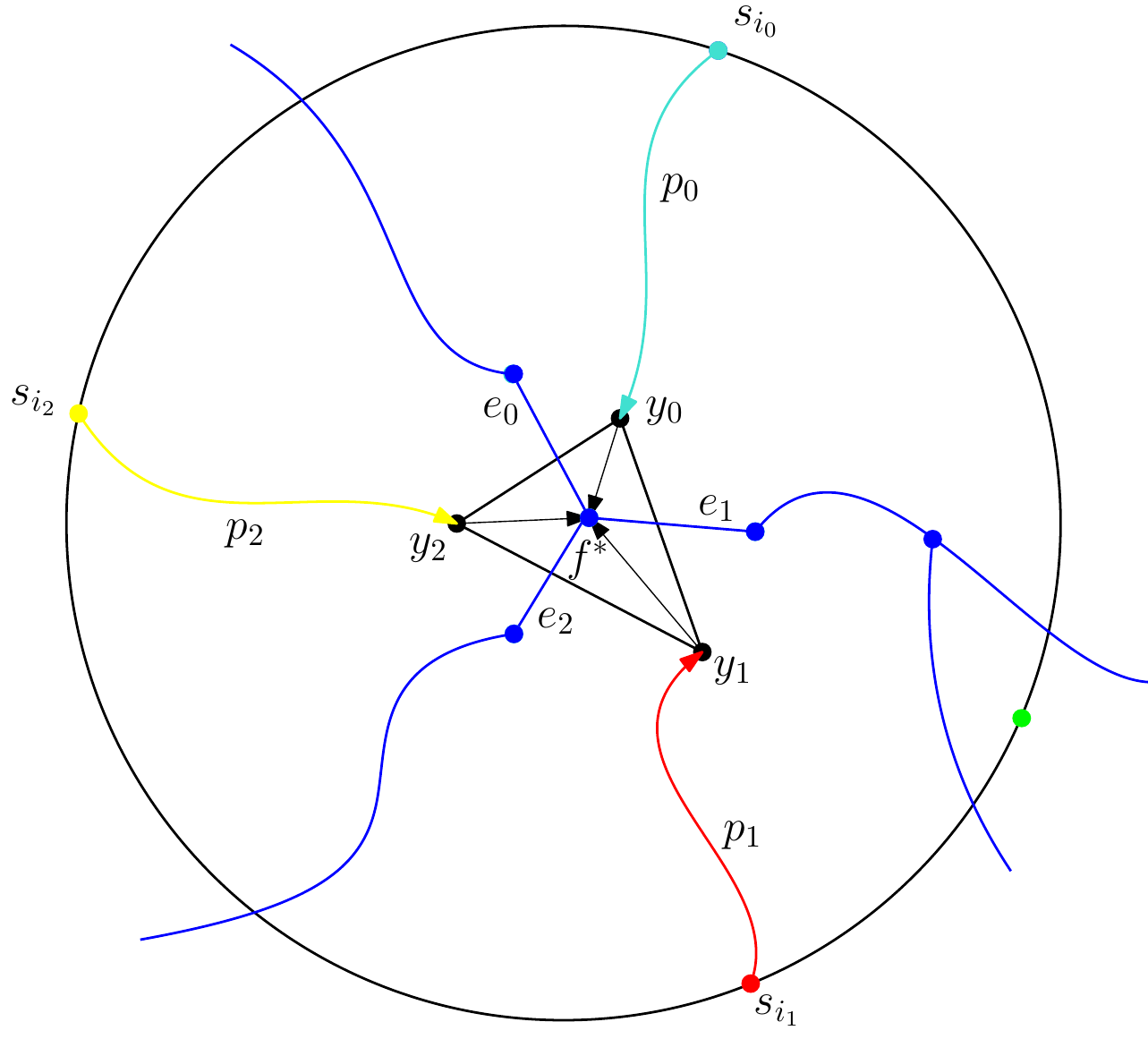}
\end{center}
\caption{Illustration of the setting and proof of Lemma~\ref{lem:location}. Left: A decomposition of $\VD^*$ (shown in blue) by a centroid $f^*$ into three subtrees, and a corresponding partition of $P$ into three regions delimited by the paths $p_i$ (shown in red, yellow, and turquoise).  Right: a schematic illustration of the same scenario.
\label{fig:schematic}
}
\end{figure}

\begin{lemma}\label{lem:location}
Let $s$ be the site such that $v\in \Vor(s)$. If $T^*$ contains all the edges of $\VD^*$ incident to $\Vor(s)$, and 
if $v$ is closer to site $s_{i_j}$ than to sites $s_{i_{j-1}},s_{i_{j+1}}$ (indices are modulo 3), then one of the following is true:
\begin{itemize}
\item $s=s_{i_{j}}$,
\item $v$ is to the right of $p_j$ and all the boundary edges of $\Vor(s)$ are contained in $T^*_j$,
\item $v$ is to the left of $p_j$ and all the boundary edges of $\Vor(s)$ are contained in $T^*_{j+1}$.
\end{itemize}
\end{lemma}
\begin{proof}
In the following, let $rev(q)$ denote the reverse of a path $q$. See Figure~\ref{fig:schematic} for an illustration of  the proof.

Let $p$ be the shortest path from $s_{i_{j}}$ to $v$. If $p$ is a subpath of $p_{j}$, then $s=s_{i_j}$.
Assume that $p$ emanates right of $p_j$ (the other case is symmetric).
First observe that the path consisting of the concatenation $p_j \circ rev(p_{j-1})$ intersects $\VD^*$ only at $f^*$. This is because, apart from the artificial arc $y_jv_f$, each shortest path $p_j$ is entirely contained in the Voronoi cell of $s_j$. Therefore, none of the subtrees $T^*_{j'}$ contains an edge dual to $p_j \circ rev(p_{j-1})$. Since the path $p_j \circ rev(p_{j-1})$ starts on $h$, ends on $h$ and contains no other vertices of $h$, it partitions the embedding into two subgraphs, one to the right of $p_j \circ rev(p_{j-1})$, and the other to its left. Since $e^*_j$ is the only edge of $T^*$ that emanates right of $p_j \circ rev(p_{j-1})$, the only edges of $T^*$ in the right subgraph are those of $T^*_j$. 

Next observe that $p$ does not cross $p_j$ (since shortest paths from the same source do not cross), and does not cross $p_{j-1}$ (since $v$ is closer to $s_{i_j}$ than to $s_{i_{j-1}}$). Since we assumed $p$ emanates right of $p_j$, the only edges of $T^*$ whose duals belong to $p$ are edges of $T^*_j$. Consider the last edge $e^*$ of $p$ that is not strictly in $\Vor(s)$. If $e^*$ does not exist then $p$ consists only of edges of $\Vor(s_{i_{j}})$, so $s=s_{i_{j}}$.  If $e^*$ does exist then it is an incident to $\Vor(s)$. By the statement of the lemma all edges of $\VD^*$ incident to  $\Vor(s)$ are in $T^*$. Therefore, by the discussion above, $e^* \in T^*_j$. We have established that some edge of $\VD^*$ incident to  $\Vor(s)$ is in $T^*_j$. It remains to show that all such edges are in $T^*_j$. The only two Voronoi cells that are partitioned by the path $p_j \circ rev(p_{j-1})$ are $\Vor(s_{i_{j}})$ and $\Vor(s_{i_{j-1}})$. Since $v$ is closer to $s_{i_{j}}$ than to $s_{i_{j-1}}$, $s \neq s_{i_{j-1}}$. Hence either $s=s_{i_{j}}$, or all the edges of $\VD^*$ incident to $\Vor(s)$ are in $T^*_j$.
\end{proof}

We can finally state and analyze the query algorithm.
We have already argued that, to locate the Voronoi cell $\Vor(s)$ to which $v$ belongs, it suffices to show how to find an edge $e^*$ incident to $\Vor(s)$.  
We start with the tree $T^{*}=\VD^*$ which trivially contains all edges of $\VD^*$ incident to $\Vor(s)$.  
We use the notation from Lemma~\ref{lem:location}. Note that we can determine in constant time which of the three sites $s_{i_j}$ is closest to $v$ by explicitly comparing the distances stored in the shortest path trees $T^{*}_{j}$. We use Proposition~\ref{lem:preorder} to determine, in constant time, whether $v$ is right of $p_j$, left of $p_j$, or a node on $p_j$. In the latter case, by Lemma~\ref{lem:location}, we can immediately infer that $v$ is in the Voronoi cell of $s_{i_j}$. In the former two cases we recurse on the appropriate subtree containing all the edges of $\VD^*$ incident to $\Vor(s)$. 
The total time is dominated by the depth of the centroid decomposition, which is $O(\log |S|)$.

\section{The Tradeoff}\label{sec:tradeoff}
In this section we generalize the construction presented in Section~\ref{sec:overall} to yield a smooth tradeoff between space and query-time. In the following, an MSSP data structure refers to Klein's multiple-source shortest paths data structure~\cite{K05}.
From now on we assume that all arc-lengths are non-negative. This can be ensured with a standard transformation
that computes shortest paths from a designated source node in $O(n\log^{2}n)$ time~\cite{KleinMW10}
and then appropriately modifies all lengths to make them non-negative while keeping the same shortest paths.

\subsection{Preprocessing}
The data structure achieving the tradeoff is recursive using Jordan curve separators as described in Section~\ref{sec:overall};
at each recursive level we have a piece $R=(V_{R},E_{R})$, which is decomposed by a Jordan curve separator $C$ into $P=(V_{P},E_{P})$
and $Q=(V_{Q},E_{Q})$, where $C$ is chosen to balance the number of nodes, the number of boundary nodes, or the number of
holes, depending on the remainder modulo 3 of the recursive level. The main difference compared to the oracle of Section~\ref{sec:overall} is that
we do not store an additively weighted Voronoi diagram of $P$ for
each node $u$ in $Q$ (and similarly we do not store a diagram of $Q$ for each node of $P$). Instead,
we use an $r$-division to decrease the number of stored Voronoi
diagrams by a factor of $\sqrt{r}$. Additionally, we stop the decomposition when the number of vertices drops below $r$.
More specifically, for every non-terminal piece $R$ in the recursive decomposition such that $n(R)>r$ 
that is decomposed into $P$ and $Q$ with a Jordan curve separator $C$, we store the following:

\begin{enumerate}
\item \label{item:MSSP} For each hole $h$ of $P$, an MSSP data structure capturing the distances in $P$ from all the boundary nodes of $P$
incident to $h$ to all nodes of $P$. The MSSP data structure is augmented with predecessor and preorder information (see below).
\item \label{item:rdiv} An $r$-division for $Q$, denoted $D_Q$, with $O(1)$ MSSP data structures for each piece of $D_Q$, one for each hole of the piece. All the boundary nodes of $Q$ (in particular, all nodes of $C$) are considered as boundary nodes of $D_Q$ (see below).
\item \label{item:dist} For each boundary node $u$ of $D_Q$, and for each boundary node $v$ of $P$,
the distance $d_G(u,v)$ from $u$ to $v$ in $G$, and also the distance $d_G(v,u)$ from $v$ to $u$ in $G$.
\item \label{item:Vor} For each boundary node $u$ of $D_Q$, and for each hole $h$ of $P$,
an additively weighted Voronoi diagram $\VD(S_h,\weight)$ for $P$, where the set of sites $S_h$ is the set of boundary nodes of $P$ incident to the hole $h$, and the additive weights $\weight$ correspond to the distances in $G$ from $u$ to each site in $S_h$. We enhance each Voronoi diagram with the point location data structure of Theorem~\ref{thm:extended}.
\end{enumerate}

We also store the same information with the roles of $Q$ and $P$ exchanged.
For a terminal piece $R$, i.e. when $n(R)\leq r$, instead of further subdividing $R$ we revert to the oracle
of Fakcharoenphol and Rao~\cite{FR06}, which needs
$O(n(R)\log n(R))$ space, answers a query in $O(\sqrt{n(R)}\log^{2}n(R))$ time, and can be constructed in $O(n(R)\log^{2}n(R))$ time.
We also construct an $r$-division $D_R$ for $R$ together with the MSSP data structures. The boundary nodes of $R$ are considered as boundary nodes of $D_{R}$.
For each boundary node $u\in \partial D_R$, and for each boundary node $v\in R$,
we store the distance $d_{G}(u,v)$ from $u$ to $v$ in $G$ and, for each hole of $R$, an enhanced additively weighted Voronoi diagram
$\VD(S_h,\weight)$ for $R$, where the set of sites $S_h$ is the set of boundary nodes of $R$ incident to the hole $h$, and the additive weights $\weight$ correspond to the distances in $G$ from $u$ to each site in $S_h$.

The MSSP data structure in item~\ref{item:MSSP} is a modification of the standard MSSP of Klein~\cite{K05}, where we change the interface
of the persistent dynamic tree representing the shortest path tree rooted at the boundary nodes of $R$ incident to $h$,
as stated by the following lemma whose proof is in the appendix.

\begin{restatable}{lemma}{mssp}
\label{lem:mssp}
Consider a directed planar embedded graph on $n$ nodes with non-negative arc-lengths, and let $v_{1},v_{2},\ldots,v_{s}$
be the nodes on the boundary of its infinite face, in clockwise order. Then, in $O(n\log n)$ time and space, we can construct
a representation of all shortest path trees $T_{i}$ rooted at $v_{i}$, that allow answering the following queries
in $O(\log n)$ time:
\begin{itemize}
	\item for a vertex $v_{i}$ and a vertex $v \in V$, return the length of the $v_{i}$-to-$v$ path in $T_{i}$.
	\item for a vertex $v_{i}$ and vertices $u,v \in V$, return whether $u$ is an ancestor of $v$ in $T_{i}$.
	\item for a vertex $v_{i}$ and vertices $u,v \in V$, return whether $u$ occurs before $v$ in the preorder traversal of $T_{i}$.
\end{itemize}
\end{restatable}

In the $r$-division in item~\ref{item:rdiv} we extend the set of boundary nodes $\partial D_Q$ of $D_Q$ to also include all
the boundary nodes of $Q$. In more detail, $D_Q$ is obtained from $Q$ by the same recursive decomposition process as the one
used to partition $G$; on every level of the recursive decomposition we choose a Jordan curve separator as to balance
the total number of nodes, boundary nodes, or holes, depending on the remainder of the level modulo 3,
and terminate the recursion when the number of nodes in a piece is $O(r)$ and the number of boundary nodes is $O(\sqrt{r})$.
Every piece of $D_Q$ consists of $O(r)$ nodes and, because the boundary nodes of $Q$ are incident to $O(1)$ holes,
its $O(\sqrt{r})$ boundary nodes are incident to $O(1)$ holes. Because of the boundary nodes inherited from $Q$, the number of pieces in $D_Q$ is not $O(n(Q)/r)$, and $|\partial D_Q|$
is not $O(n(Q)/\sqrt{r})$. We will analyze $|\partial D_Q|$ later.
The MSSP data structures in item~\ref{item:rdiv} stored for every piece of $D_Q$ are the standard structures of Klein.
The distances in item~\ref{item:dist} are stored explicitly.
The point location mechanism used for the Voronoi diagrams in item~\ref{item:Vor} is the one described
in Section~\ref{sec:pointlocation} with the following important modification. Instead of storing the shortest path
trees rooted at every site of the Voronoi diagram explicitly to report distances, preorder numbers, and ancestry relations
in $O(1)$ time, use the MSSP data structure stored in item~\ref{item:MSSP}.
Clearly, with such queries one can implement Proposition~\ref{lem:preorder} in $O(\log(|P|))$ time instead of $O(1)$.

\subsection{Query}

To compute the distance from $u$
to $v$, we traverse the recursive decomposition starting from the piece that corresponds to the
whole initial graph $G$ as in Section~\ref{sec:overall}. Eventually, we reach a piece $R=(V_{R},E_{R})$ such that
$u,v\in V_{R}$ and either $n(R)\leq r$, or $n(R)>r$ and $R$ is decomposed into $P$ and $Q$ with a Jordan curve
separator $C$ such that either $u\in C$, or $v\in C$, or $u$ and $v$ are separated by $C$.

We first consider the case when $n(R)>r$.
If $u\in C$ or $v\in C$ then, because the nodes of $C$ are boundary nodes the
$r$-divisions $D_P$ and $D_Q$, the distance from $u$ to $v$ in $G$ can be extracted from item~\ref{item:dist}.
Otherwise, $u$ and $v$ are separated by $C$, and we assume without loss of generality that $u\in Q$
and $v\in P$. 
Let $Q'$ be the piece of the $r$-division $D_Q$ that contains $u$.
Any path from $u$ to $v$ must visit a boundary node of $Q'$. Thus, we can iterate over the
boundary nodes $u'$ of $Q'$, retrieve $d_{Q'}(u,u')$ (from the MSSP data structure in item~\ref{item:rdiv}),
and then, for each hole $h$ of $P$, use the Voronoi diagram $\VD(S_h,\weight)$ for $P$ (item~\ref{item:Vor})
to find the node $v'\in S_{h}$ that minimizes
$d_{G}(u',v')+d_{P}(v',v)$ (computed from item~\ref{item:dist} and item~\ref{item:MSSP}). The
minimum value of $d_{Q'}(u,u')+d_{G}(u',v')+d_{P}(v',v)$ found during
this computation corresponds to the shortest path from $u$ to $v$.

The remaining possibility is that $n(R)\leq r$. Then the shortest path from $u$ to $v$ either
visits some boundary node of $R$ or not. To check the former case, we proceed similarly as above:
we find the piece $R'$ of the $r$-division $D_R$ that contains $u$, iterate over the
boundary nodes $u'$ of $Q'$, retrieve $d_{R'}(u,u')$, and use the Voronoi diagram
$\VD(S_h,\weight)$ for $R$ to find the node $v'\in S_{h}$ that minimizes
$d_{G}(u',v')+d_{P}(v',v)$. To check the latter case, we query the oracle of
Fakcharoenphol and Rao~\cite{FR06} stored for $R$, and return the minimum
of these two distances.

\subsection{Analysis}

For a piece $R$, we denote by $n(R)$ and $b(R)$ the number of nodes and boundary nodes of $R$, respectively.
We first analyze the query-time. In $O(\log n)$ time we reach the appropriate piece $R$. 
Then, we iterate over $O(\sqrt{r})$ boundary nodes. For each of them, we first spend
$O(\log r)$ time to retrieve the distance from $u$ to $u'$. Then, we need $O(\log (b(P)) \log (n(P)))$
time to query the Voronoi diagram. If $n(R)\leq r$, this changes into $O(\log(b(R))\log (n(R)))$
and additional $O(\sqrt{n(R)}\log^{2}(n(R)))=O(\sqrt{r}\log^{2}r)$ time for the oracle of
Fakcharoenphol and Rao~\cite{FR06}. Thus, the total query-time is $O(\sqrt{r}\log^{2}n)$. 

We bound the space required by the data structure for a piece $R$ which is divided into pieces $P$ and $Q$.
Each MSSP data structure in item~\ref{item:MSSP} requires $O(n(P)\log(n(P)))$ space, and there are $O(1)$ of them.
Representing the $r$-division $D_Q$ and the MSSP data structures for all the pieces in item~\ref{item:rdiv} can be done
within $O(n(Q)\log r)$ space. Then, for every boundary node of $D_Q$ the distances in item~\ref{item:dist}
and the $O(1)$ Voronoi diagrams in item~\ref{item:Vor} can be stored in $O(b(P))$ 
space.
Thus, we need to analyze the total number of boundary nodes of $D$. As we explained above, $|\partial D_Q|$ would be
simply $O(n(Q)/r)$ if not for the additional boundary nodes of $Q$. We claim
that $|\partial D_Q| = O(n(Q)/\sqrt{r}+b(Q))$.

To prove the claim we slightly modify the reasoning used by Klein, Mozes, and Sommer~\cite{KMS13} to bound the total number of boundary nodes in an $r$-division without additional boundary vertices.
They analyzed the same recursive decomposition process of a planar graph $Q$ on $n$ nodes
by separating to balance the number of nodes, boundary nodes, or holes, depending on the remainder
modulo 3 of the current level.\footnote{In~\cite{KMS13} simple cycle separators are used (rather than Jordan curve separators), and thus every piece along the recursion needs to be re-triangulated. The analysis of the number of boundary nodes, however, is the same.}
Let $\mathcal{T}$ be a tree representing this process, and $\hat x$
be the root of $\mathcal{T}$. Every node $x$ of $\mathcal{T}$ corresponds to a piece. For example, the piece corresponding to the root $\hat x$ is all of $Q$. We denote by $n(x)$ and $b(x)$ the number of nodes and boundary nodes, respectively, of the piece corresponding to $x$. 
Define $S_{r}$ to be the set of rootmost nodes $y$ of $\mathcal{T}$
such that $n(y)\leq r$.
\begin{lemma} \label{lem:Sr}
$\sum_{x \in S_r} b(x) = O(n(Q)/\sqrt r + b(Q))$	
\end{lemma}
\begin{proof}
For a node $x$ of $\mathcal{T}$
and a set $S$ of descendants of $x$ such that no node of $S$ is an ancestor of any other, define
$L(x,S):=-n(x)+\sum_{y\in S}n(y)$. 
Essentially, $L(x,S)$ counts the number of new boundary nodes with multiplicities created when replacing $x$ by all pieces in $S$.
Lemma 8 in~\cite{KMS13} states that $L(\hat x,S_{r}) = O(n/\sqrt{r})$. I.e., the number of new boundary nodes (with multiplicities) created when replacing the single piece $Q$ by the pieces in $S_r$ is $O(n(Q)/\sqrt{r})$.
We  assume that each node of $Q$ has constant degree (this can be guaranteed with
a standard transformation). Thus, each boundary vertex of $Q$ appears in a constant number of pieces in $S_r$.
Since the number of boundary vertices in $Q$ is $b(Q)$,
 the lemma follows.
\end{proof}

Let $S'_{r}(x)$ be the set of rootmost descendants $y$ of $x$ such that $b(y)\leq c'\sqrt{r}$, where $c'$ is a fixed known constant.
The $r$-division found by the recursive decomposition
process is $S'_{r}=\bigcup_{x\in S_{r}} S'_{r}(x)$. Indeed, 
each piece $x$ in $S'_r$
has $n(x)\leq r$, $b(x) \leq c'\sqrt{r}$, and $O(1)$ holes. This is true by definition of $S'_r$, even though, instead of starting with a graph with no boundary nodes, we start with a graph
containing $b(Q)$ boundary nodes incident to $O(1)$ holes.

The following claim is proved in  
Lemma 9 of~\cite{KMS13}.
\begin{lemma}[Lemma 9 of~\cite{KMS13}] \label{lem:Srt}
$|S'_r(x)| \leq \max\{1, \frac{40b(x)}{c'\sqrt{r}}\}$	
\end{lemma}

\begin{corollary}
$|\partial D_Q| = O(n(Q)/\sqrt r + b(Q))$	
\end{corollary}
\begin{proof}
	$$ |\partial D_Q| \leq \sum_{x\in S'_r}b(x) \leq c'\sqrt{r} \sum_{x \in S_r} |S'_r(x)|  \leq c'\sqrt{r} \left( |S_r| + \frac{40}{c'\sqrt r}\sum_{s\in S_r}b(x) \right) = O(n(Q)/\sqrt r + b(Q)).$$
Here, the first inequality follows by definition of $D_Q$ and $S'_r$. The second inequality follows by definition of $S'_r(x)$ and by the fact that for any $x \in S'_r$, $b(x)\leq c'\sqrt r$. The third inequality follows by Lemma~\ref{lem:Srt}. The last inequality follows from the fact that $|S_r| = O(n(Q)/r)$, and from Lemma~\ref{lem:Sr}.
\end{proof}

We have shown that, starting the recursive decomposition of $Q$ with $b(Q)$ boundary nodes
incident to $O(1)$ holes, we obtain an $r$-division $D$ consisting of pieces containing
$O(r)$ nodes and $O(\sqrt{r})$ boundary nodes incident to $O(1)$ holes, and $O(n(Q)/\sqrt{r}+b(Q))$ boundary
nodes overall. Consequently, the space required by the data structure for a piece $R$ with $n(R)>r$ that is separated into pieces $P$ and $Q$ is
$O(n(P)\log(n(P)+n(Q)\log r + n(Q)/\sqrt{r}+b(Q))b(P))$, plus a symmetric term with the roles of $P$ and $Q$ exchanged.
If $n(R)\leq r$ the space is $O(n(R)\log(n(R))+(n(R)/\sqrt{r}+b(R))b(R))$.
Overall, $O(n(R)\log(n(R))$ sums up $O(n\log n\log(n/r))$. On the $\ell$-th level of the decomposition,
$O(n(P)/\sqrt{r}\cdot b(Q)+n(Q)/\sqrt{r}\cdot b(P))$ sums up to $O(n^{1.5}/(\sqrt{r}\cdot c^{\ell}))$, so $O(n^{1.5}/\sqrt{r})$ over all levels.
$O(b(P)\cdot b(Q))$ can be bounded by $O(b(Q)\cdot \sqrt{n})$, so we only need to bound
the total number of boundary nodes in all pieces of the recursive decomposition
of the whole graph. For the terminal pieces $R$, it directly follows from
Lemma~\ref{lem:Sr} (with $Q$ being the entire graph $G$) that the total number of boundary nodes
is $O(n/\sqrt{r})$, but we also need to analyse the non-terminal pieces.
Because the size of a piece decreases by a constant factor after
at most three steps of the recursive decomposition process, it suffices to bound only the total number of boundary nodes for pieces in the sets $S_{r_i}$, for $r_i = r\cdot 2^i$, $i =0,1,\ldots,\log(n/r)$. By applying Lemma~\ref{lem:Sr}, with $r=r_i$ and $Q=G$ we get that the total number of boundary nodes for pieces in $S_{r_i}$ is $O(n/\sqrt{r_i})$, which sums up to $O(n/\sqrt r)$ over all $i$. 
Thus, the sum of $O(b(P)\cdot b(Q))$
over all non-terminal pieces $O(n^{1.5}/\sqrt{r})$.
For all terminal pieces $R$, $O(n(R)/\sqrt{r}\cdot b(R))$ adds up to $O(n^{1.5}/\sqrt{r})$. $O(b(R)\cdot b(R))$
can be bounded by $O(b(R)\cdot \sqrt{n})$, which we have already shown to be $O(n^{1.5}/\sqrt{r})$ overall.
Thus, the total space is $O(n^{1.5}/\sqrt{r}+n\log n\log(n/r))$.

The preprocessing time can be analyzed similarly as in Section~\ref{sec:overall}, except that now we need
to compute only $O(n(P)/\sqrt{r}+b(P))$ Voronoi diagrams for $P$, each in $\tilde O(n(P))$ time.
As shown above, the overall number of boundary nodes is $O(n/\sqrt{r})$, so this is
$\tilde O(n^{2}/\sqrt{r})$ total time.
Additionally, we need to compute
the distance between pairs of vertices of $P$ in $G$ (item~\ref{item:dist}). One of these vertices is always a boundary node of
the $r$-division, so overall we need $\tilde O(n/\sqrt{r})$ single-source shortest paths computations in $G$,
which takes $\tilde O(n^{2}/\sqrt{r})$ total time.
Additionally, we need to construct the oracles when $n(R)\leq r$
in $O(n\log^{2}r)$ total time. Thus, the total construction time is $\tilde O(n^{2}/\sqrt{r})$ overall.

\subsection{Improved query-time}
The final step in this section is to replace $\log^{2} n$ with $\log n \cdot \log r$ in the query-time. This is done by
observing that the augmented MSSP data structure takes linear space, but for smaller values of $r$ we can actually
afford to store more data. In the appendix we show the following lemma.

\begin{restatable}{lemma}{mssp2}
\label{lem:mssp2}
For any $r\in [1,n]$, the representation from Lemma~\ref{lem:mssp} can be modified to allow answering queries
in $O(\log r)$ time in $O(s\cdot n/\sqrt{r}+n\log r)$ space after $O(s\cdot n/\sqrt{r}\log r+n\log n)$ time preprocessing.
\end{restatable}

This decreases the query-time to $\Oh(\sqrt{r}\log n\log r)$ at the expense of increasing the space taken by the
MSSP data structures in item~\ref{item:MSSP} to $\Oh(n(P)/\sqrt{r}\cdot b(P)+n(P)\log r)$. Summing over all levels $\ell$
and including the space used by all other ingredients, this is $\Oh(n^{1.5} / \sqrt{r} + n\log r\log (n/r))$.

\section{Removing the Assumption on Sites}
\label{sec:extension}

We now remove the assumption that all the vertices on the hole $h$ are sites of the Voronoi diagram whose Voronoi cells are non-empty.
Recall that $\VD^*(S,\weight)$ is obtained from $VD_1^*$ by replacing $h^*$ with multiple copies, one copy $h^*_e$ for each edge $e$ of $\VD_0^*(S,\weight)$ incident to $h^*$. Consider the proof of Lemma~\ref{lem:tree}. The argument showing that $\VD^*(S,\weight)$ contains no cycles still holds. However, because now vertices incident to the hole are not necessarily sites, the argument showing connectivity fails, and indeed, $\VD^*(S,\weight)$ might be a forest. We turn the forest $\VD^*$ into a tree $\widehat{\VD}^*$ by identifying certain pairs of copies of $h^*$ as follows.
 Consider the sequence $E_h$ of edges of $\VD^*$ incident to $h^*$, ordered according to their clockwise order on the face $h$. Each pair of consecutive edges $e,e'$ in $E_h$ delimits a subpath $Q$ of the boundary walk of $h$. Note that  $Q$ belongs to a single Voronoi cell of some site $s\in S$ . If $Q$ does not contain $s$ we connect the two copies $h^*_e$ and $h^*_{e'}$ with an artificial edge. We denote the resulting graph by $\widehat{\VD}^*$. See Figure~\ref{fig:generaltree}.
 
\begin{figure}[h]
\begin{center}
\includegraphics[width=0.47\textwidth]{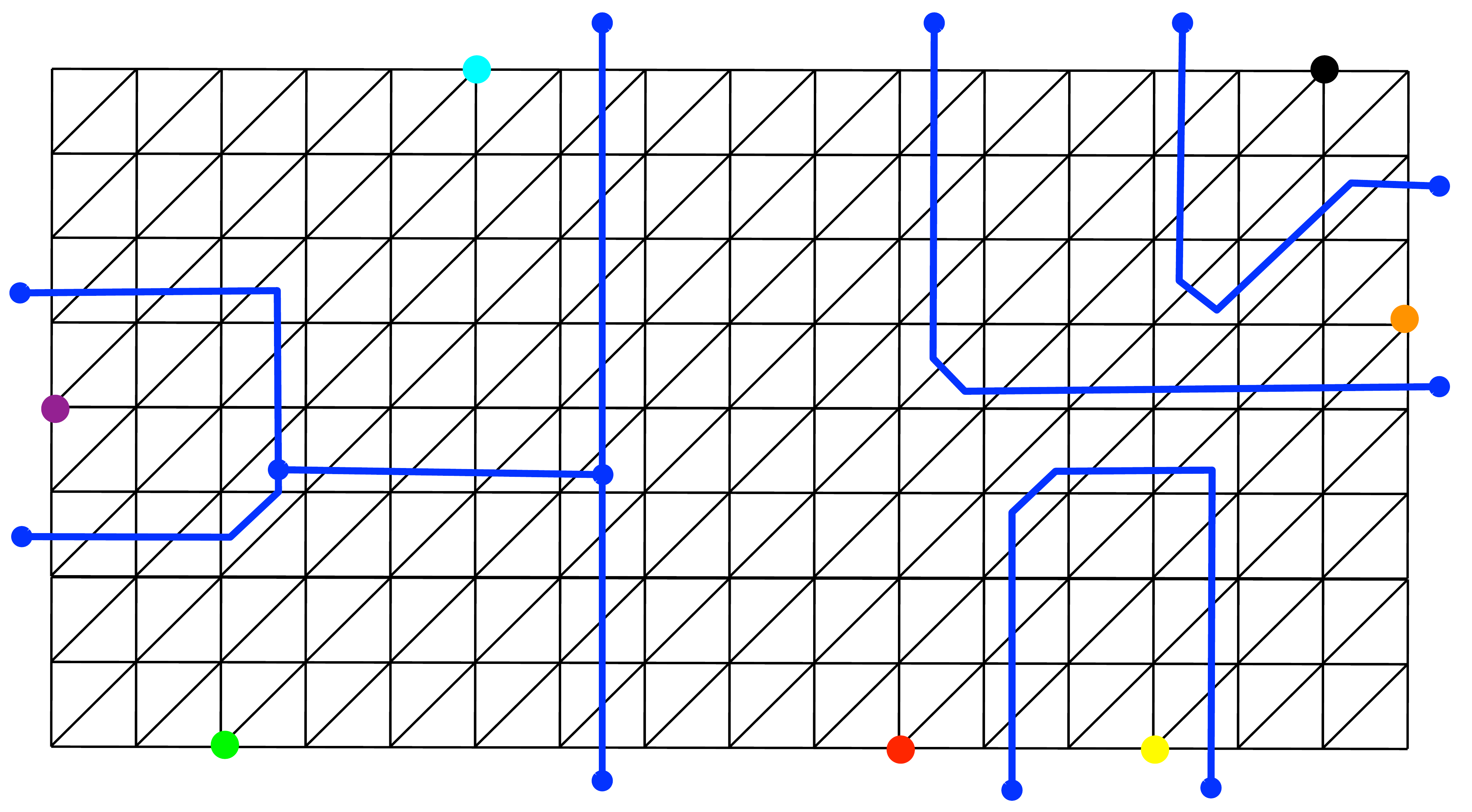}
\includegraphics[width=0.47\textwidth]{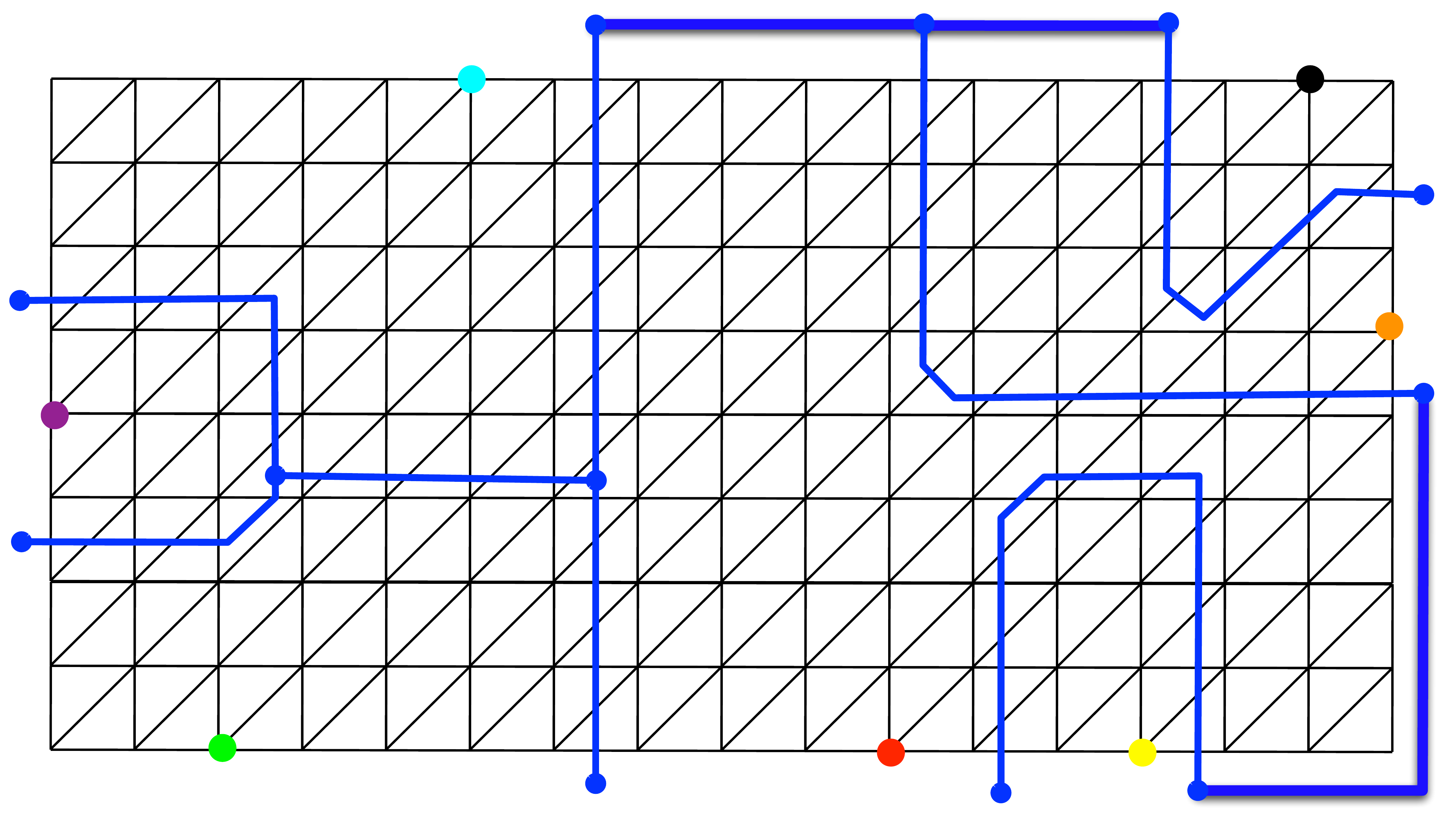}
\end{center}
\caption{Left: a Voronoi diagram $\VD^*$ (blue) forms a forest. Right: the tree $\widehat{\VD}^*$ obtained by adding three artificial edges (thicker blue lines). \label{fig:generaltree}}
\end{figure} 
 \begin{lemma}\label{lem:hattree}
 $\widehat{\VD}^*$ is a tree.	
 \end{lemma}

\begin{proof}
 
 We show that $\widehat{\VD}^*$ is connected and has no cycles. Consider the Voronoi cell of a site $s \in S$. In $\VD_1^*$ (i.e., before splitting $h^*$) the boundary of this cell is a non-self-crossing cycle $C^*$. Consider the restriction of $C$ to edges incident to $h^*$. Consider two consecutive edges $e,e'$ in the restriction. 
 If $e$ and $e'$ are not consecutive in $C^*$ (i.e., if $e$ and $e'$ do not meet at $h^*$), then they remain connected in $\VD^*$ (i.e., after splitting $h^*$). If $e$ and $e'$ are consecutive on $C^*$, then they are also consecutive in $E_h$, so they become disconnected in $\VD^*$, but get connected again in $\widehat{\VD}^*$. This is true unless the subpath $Q$ of the face $h$ delimited by $e$ and $e'$ contains $s$, but this only happens for one pair of edges in $C^*$. Therefore, since $C^*$ was 2-connected (a cycle) in $\VD_1^*$, it is 1-connected in $\widehat{\VD}^*$. Now, since adjacent Voronoi cells share edges, the boundaries of any two adjacent Voronoi cells are connected. It follows from the fact that the dual graph of $\widehat{\VD}^*$ is connected that the boundaries of all cells are connected after the identification step.
 
Assume that $\widehat{\VD}^*$ contains a cycle $C^*$. Then $C^*$ must also be a cycle in $\VD^*_1$. Since every cycle in $\VD_1^*$ contains $h^*$ and encloses at least one site, $C^*$ contains a copy of $h^*$ and encloses at least one site. Consider a decomposition of $C^*$ into maximal segments between copies of $h^*$. By construction of $\widehat{\VD}^*$, whenever two segments of $C^*$ are connected with an artificial edge, the segment of the boundary of $h$ delimited by these two segments and enclosed by $C^*$ does not contain a site. Since all the sites are on the boundary of $h$, it follows that $C^*$ does not enclose any sites, a contradiction.
\end{proof}  

We next describe how to extend the point location data structure. 
First observe that since $\widehat{\VD}^*$ is a tree, it has a centroid decomposition. 
In $\VD^*$, the copies of $h^*$ are all  leaves. In $\widehat{\VD}^*$ each copy of $h^*$ is incident to at most two artificial edges. Hence the maximum degree in $\widehat{\VD}^*$ is still 3. 
If the centroid of $\widehat{\VD}^*$ is not a copy of $h^*$ then Lemma~\ref{lem:location} holds. We need a version of Lemma~\ref{lem:location} for the case when the centroid is a copy of $h^*$ with degree greater than 1 (i.e., incident to one or two artificial edges). This is in fact a simpler version of Lemma~\ref{lem:location}. The difference between a copy of $h^*$ and a Voronoi vertex $f^*$ is that $f^*$ is a triangular face incident to three specific vertices $y_0,y_1,y_2$, whereas $h^*$ is incident to all vertices of the hole $h$.  Recall that we connect two copies $h^*_e$ and $h^*_{e'}$ if the segment $Q$ of the boundary of $h$ delimited by the edges $e$ and $e'$ belongs to the Voronoi cell of a site $s$ but does not contain $s$. When we add this artificial edge, we associate with $h^*_e$ and with $h^*_{e'}$ an arbitrary primal vertex $y$ on $Q$. 
Thus, each copy $f$ of $h^*$ is associated with at most two primal vertices. 

We describe the case where the centroid of $T^*$ is a copy $\hat h^*$ of $h^*$ with degree 3. 
The case of degree 2 and one associated vertex is similar.
In the case of degree 3, $\hat h^*$ is incident to one edge $e_1 \in E_h$, and to two artificial edges which we denote $e_0$, and $e_2$, so that the counterclockwise order of edges around $\hat h^*$ is $e_0,e_1,e_2$.
Removing $\hat h^*$ breaks $T^*$ into three subtrees. Let $T^*_j$ be the subtree of $T^*$ rooted at the endpoint of $e_j$ that is not $\hat h^*$.  
Recall that, since the degree of $\hat h^*$ is 3, it  has two associated vertices, $y_0, y_1$, where $y_j$ belongs to the subpath of the boundary of $h$ delimited by $e_j$ and $e_{j+1}$.
 Let $s_{i_j}$ be the site such that $y_j$ belongs to $\Vor(s_{i_j})$. Let $p_j$ be the shortest path from $s_{i_j}$ to $y_j$. See Figure~\ref{fig:generalloc}.
 
\begin{figure}[h]
\begin{center}
\includegraphics[width=0.7\textwidth]{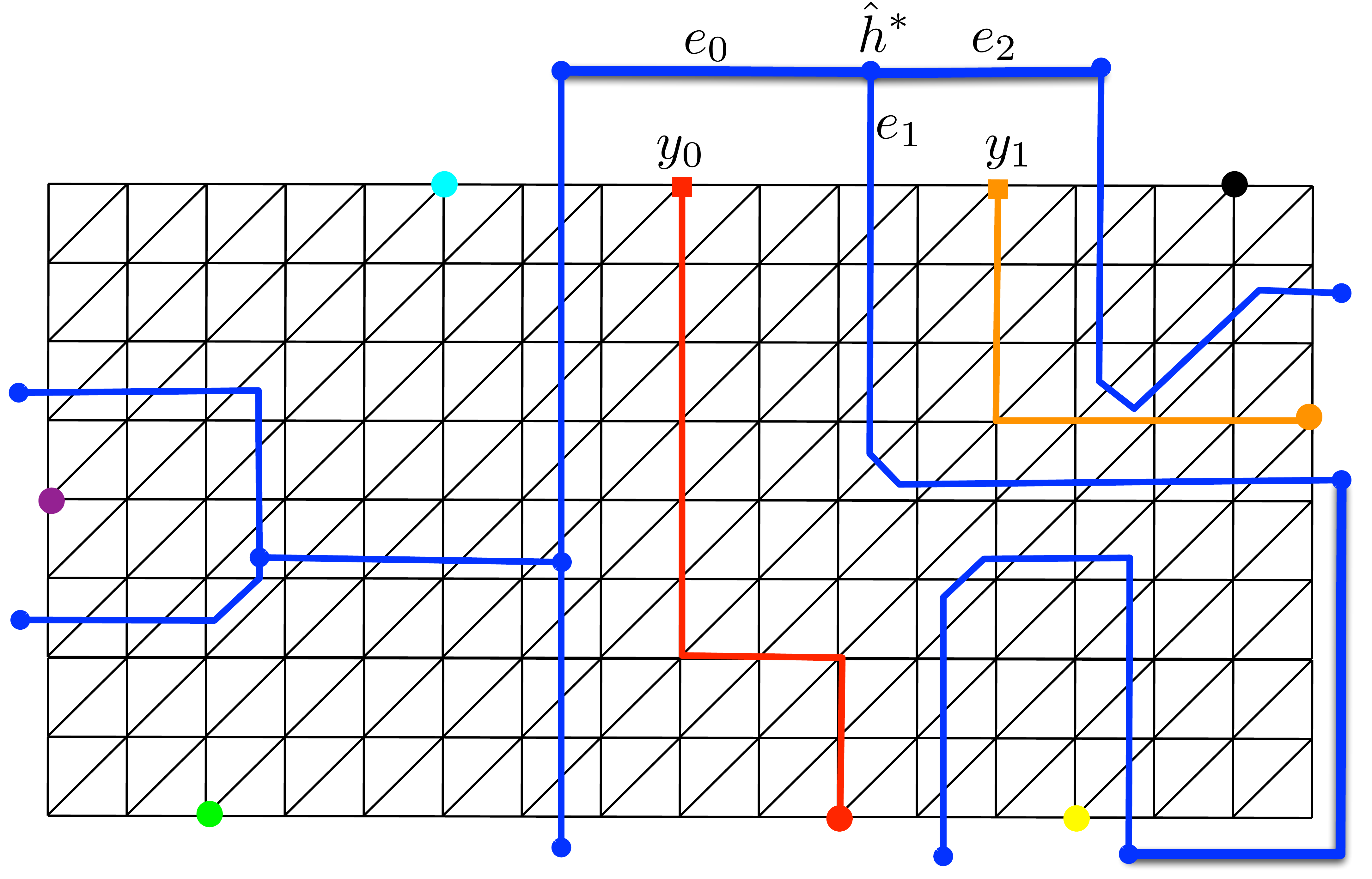}
\end{center}
\caption{Illustration of the case when the centroid is a copy $\hat h^*$ of $h^*$ of degree three. $\hat h^*$ has two incident artificial edges ($e_0,e_2$), and one Voronoi edge ($e_1$). The two vertices $y_0,y_1$ associated with $\hat h^*$ are shown, as well as the shortest paths $p_0$ and $p_1$. \label{fig:generalloc}}
\end{figure}

\begin{lemma}\label{lem:location2}
Let $s$ be the site such that $v\in \Vor(s)$. If $T^*$ contains all the edges of $\widehat{\VD}^*$ incident to $\Vor(s)$, and if $v$ is closer to site $s_{i_j}$ than to site $s_{i_{j+1}}$ (indices are modulo 2), then one of the following is true:
\begin{itemize}
\item $s=s_{i_{j}}$,
\item $v$ is to the left of $p_j$ and all the edges of $\widehat{\VD}^*$ incident to $\Vor(s)$ are contained in $T^*_j$,
\item $v$ is to the right of $p_j$ and all the edges of $\widehat{\VD}^*$ incident to $\Vor(s)$ are contained in $T^*_{j+1}$.	
\end{itemize}
\end{lemma}

\begin{proof}
Observe that all the edges of $p_j$ belong to $\Vor(s_{i_j})$, while for every $i\in \{0,1,2\}$, the duals of edges of $T^*_i$ have endpoints in two different Voronoi cells. Therefore, the paths $p_j$ do not cross the trees $T^*_i$. 
Since $p_0$ and $p_1$ are paths that start and end on the boundary of $h$ and do not cross each other, they partition $G$ into three subgraphs $\{G_i\}_{i=0}^2$. Let $G_0$ be the subgraph to the left of $p_0$, $G_1$ the subgraph to the right of $p_0$ and to the left pf $p_1$, and $G_2$ the graph to the right of $p_1$. 
It follows from the above that each subtree $T^*_i$ belongs to the subgraph $G_i$.

The remainder of the proof is almost identical to that of Lemma~\ref{lem:location}.
Let $p$ be the shortest path from $s_{i_j}$ to $v$. If $p$ is a subpath of $p_j$ then $s=s_{i_j}$. Otherwise, assume $p$ emanates left of $p_0$ (the other cases are similar). Consider the last edge $e^*$ of $p$ that is not strictly in $\Vor(s)$. If $e^*$ does not exist then $s=s_{i_0}$. If it does exist, then it must be en edge of $T^*_0$. Since the only Voronoi cell partitioned by $p_0$ is that of $s_{i_0}$, either $s = s_{i_0}$, or all edges of $\widehat{\VD}^*$ incident to $\Vor_s$ belong to $T^*_0$.   	
\end{proof}
 
\bibliographystyle{abbrv}

\newpage
\appendix

\section{Missing Proofs}

\separating*

\begin{proof}
The number of holes increases by at most one in every recursive call, but decreases by a constant
multiplicative factor every 3 recursive calls, and is initially equal to 0, so part (i) easily follows.
The number of nodes never increases, and decreases by a constant multiplicative factor
every 3 recursive calls, and is initially equal to $n$, so part (ii) follows.
The situation with the number of boundary nodes is slightly more complex,
because it increases by $O(\sqrt{n})$ in every recursive call, and decreases
by a constant multiplicative factor every 3 recursive calls, where $n$ is
the number of nodes in the current piece.
For simplicity, we analyze a different process, in which the number of boundary
nodes decreases by a constant multiplicative factor and then increases by
$O(\sqrt{n})$ in every recursive call. The asymptotic behavior of these two
processes is identical. Thus, we want to analyze the following recurrence:
\[ b(\ell+1) = b(\ell)/c + \frac{\sqrt{n}}{(c')^{\ell}}. \]
for some constants $c,c'>1$. Then
\[ b(\ell+1) = \sum_{i=0}^{\ell} \frac{\sqrt{n}}{c^{i}(c')^{\ell-i}} = \frac{\sqrt{n}}{(c')^{\ell}} \sum_{i=0}^{\ell} (\frac{c'}{c})^{i}. \]
We consider two cases:
\begin{enumerate}
\item $c' \leq c$, then $b(\ell+1) \leq \sqrt{n} \frac{\ell}{(c')^{\ell}} \leq \frac{\sqrt{n}}{(\sqrt{c'})^{\ell}}$ for $\ell$ large enough.
\item $c' > c$, then $b(\ell+1) = O(\frac{\sqrt{n}}{c^{\ell}})$.
\end{enumerate}
In both cases, $b(\ell) = O(\frac{\sqrt{n}}{c_{2}^{\ell}})$ for some constant $c_{2}>1$ as claimed.
\end{proof}

\mssp*

\begin{proof}
We proceed as in the original implementation of MSSP, that is, we represent every $T_{i}$ with a persistent link-cut tree.
We present some of the details, understanding them is required to explain how to implement the queries.

In MSSP, we start with constructing $T_{1}$ with Dijkstra's algorithm in $O(n\log n)$ time. Then, we iterate
over $i=2,3,\ldots,s$. The current $T_{i}$ is maintained with a persistent link-cut tree of Sleator and Tarjan~\cite{SleatorT83}.
The gist of MSSP is that every edge of the graph goes in and out of the shortest path tree at most once,
and that we can efficiently retrieve the edges that should be removed from and added to $T_{i-1}$ to obtain
$T_{i}$ (in $O(\log n)$ time per edge). Thus, if we are able to remove or insert an edge from $T_{i-1}$ in
$O(\log n)$ time, the total update time is $O(n\log n)$. With a link-cut tree, we can indeed remove or
insert edge in such time (note that we prefer the worst-case version instead of the simpler implementation
based on splay trees). We make our link-cut tree partially persistent with a straightforward application of the
general technique of Driscoll et al.~\cite{DriscollSST89}. This requires that the in-degree of the underlying
structure is $O(1)$, which is indeed the case if the degrees of the nodes in the graph (and hence in every
$T_{i}$) are $O(1)$. This can be guaranteed by replacing a node of degree $d>4$ by a cycle on $d$ nodes,
where every node has degree $3$.  We now verify that the in-degree is $O(1)$ for such a structure
by presenting a high-level overview of link-cut trees.

The edges of a rooted tree are partitioned into solid and dashed. There is at most one solid edge incoming
into any node, so we obtain a partition of the tree into node-disjoint solid paths. For every solid path,
we maintain a balanced search tree on a set of leaves corresponding to the nodes of the path in the
natural top-bottom order when read from left to right. To obtain a worst-case time bound, Sleator and Tarjan
use biased binary trees. Every node stores a pointer to the leaf in the corresponding biased binary
tree, and additionally the topmost node of a heavy path stores a pointer to its parent in the represented
tree (together with the cost of the corresponding edge). The nodes of every biased binary tree store
standard data (a pointer to the left child, the right child, and the parent) and, additionally, every inner
node (that corresponds to a fragment of a solid path) stores the total cost of the corresponding
fragment. An additional component of the link-cut tree is a complete binary tree on $n$ leaves corresponding
to the nodes of the tree (called $1,2,\ldots,n$). This is required, so that we can access a node of
the tree on demand in $O(\log n)$ time, as random access is not allowed in this setting. The access
pointer points to the root of the complete binary tree. Now we can indeed verify that the in-degree of
the structure is $O(1)$.

Assuming that a representation of every $T_{i}$ with a partially persistent link-cut tree is available,
we can answer the queries as follows.

First, consider calculating the distance from $v_{i}$
to some $v\in V$. We retrieve the access pointer of $T_{i}$ and navigate the complete binary tree
to reach the node $v$. Then, we navigate up in the link-cut representation of $T_{i}$ starting
from $v$. In every step, we traverse a biased binary tree starting from a leaf corresponding to
some ancestor $u$ of $v$. Conceptually, this allows us to jump to the topmost node of the
solid path containing $u$. While doing so, we accumulate the total cost of the prefix of that
solid path ending at $u$. Then, we follow the pointer from the topmost node of the current
solid path to reach its parent in $T_{i}$, add its cost to the answer, and continue in the next
biased binary tree. It is easy to see that in the end we obtain the total cost of the path from
$v$ to the root of $T_{i}$, and by the properties of biased binary trees the total number of steps
is $O(\log n)$.

Second, consider checking if $u$ is an ancestor of $v$ in $T_{i}$. We navigate in the link-cut
representation of $T_{i}$ starting from $u$ and marking the visited solid paths (in more detail,
every solid path stores a timestamp of the most recent visit; the timestamps are not considered
a part of the original partially persistent structure and the current time is increased after each query).
Then, we navigate starting from $v$, but stop as soon as we reach a solid path already visited in
the previous step. For $u$ to be an ancestor of $v$, this must be the path containing $u$,
and furthermore $u$ must be on the left of $v$ in the corresponding biased binary tree.
This can be all checked in $O(\log n)$ time.

Third, consider checking if $u$ occurs before $v$ in the preorder traversal of $T_{i}$. By
proceeding as in the previous paragraph we can identify the LCA of $u$ and $v$, denoted
$w$. Assuming that $w\neq u$ and $w\neq v$, we can also retrieve the edge outgoing
from $w$ leading to the subtree containing $u$, and similarly for $v$, in $O(\log n)$
total time. We can also retrieve the edge incoming to $w$ from its parent in $O(1)$
additional time. Then, we check the cyclic order on the edges incident to $w$
in the graph to determine if $u$ comes before $v$ in the preorder traversal of $T_{i}$
(this is so that we do not need to think about an embedding of $T_{i}$ while maintaining
the link-cut representation).
\end{proof}

\mssp*

\begin{proof}
We construct an $r$-division $R$ of the graph. The structure consists of parts: micro components and macro component.

Consider a shortest path tree $T_{i}$. We construct a new (smaller) tree $T'_{i}$
as follows. First, mark in $T_{i}$ $v_{i}$ and all boundary nodes of $R$. Then, $T'_{i}$ is the tree induced by the marked nodes
in $T_{i}$ (in other words: for any two marked nodes we also mark their LCA in $T_{i}$, and then construct $T'_{i}$ by connecting
every marked node to its first marked ancestor with an edge of length equal to the total length of the path connecting them
in $T_{i}$. Then, $|T'_{i}| = O(n/\sqrt{r})$. Intuitively, $T'_{i}$ gives us a high-level overview of the whole $T_{i}$. We augment
it with the usual preorder numbers and LCA data structure. We call this the macro component.

For every piece $R_{j}$ of $R$, consider the subgraph of $T_{i}$ consisting of all edges belonging to $R_{j}$. This subgraph
is a collection of trees rooted at some of the boundary nodes of $R_{j}$. We represent this forest $T_{i,j}$ with a persistent link-cut
forest. While sweeping through the nodes $v_{1},v_{2},\ldots,v_{s}$, every edge of the graph goes in and out of the shortest
path tree at most once. Hence, every edge of $R_{j}$ goes in and out of $T_{i,j}$ at most once. Thus, the persistent link-cut
representation of $T_{i,j}$ takes $O(|R|\log |R|)$ time and space. To answer a query concerning $T_{i,j}$, we first need
to retrieve the corresponding version of the link-cut forest. This can be done with a predecessor search, if we store
a sorted list of the values of $i$ together with a pointer to the corresponding version, in $O(\log r)$ time, as there are at
most $O(|R|)$ versions. Then, a query concerning $T_{i,j}$ can be answered in $O(\log r)$ time.

We claim that combining the micro and macro components allows us to answer any query in $O(\log r)$ time.
Consider calculating the distance from $v_{i}$ to some $v\in V$. We retrieve the piece $R_{j}$ containing $v$ and,
by using the micro component, find the root $r$ of the tree containing $v$ in the forest $T_{i,j}$ together with
the distance from $r$ to $v$ in $O(\log r)$ time. Then, $r$ is a boundary node, so the macro component allows
us to find the distance from $v_{i}$ to $r$ in $O(1)$ time. Other queries can be processed similarly by first look
at the pieces containing $u$ and $v$, then replacing them by appropriate boundary nodes, and finally looking at
the macro component.

The total space is clearly $O(n\log r)$ to represent all the micro components, and $O(s\cdot n/\sqrt{r})$ for
the macro component. To bound the preprocessing time, observe that constructing the macro component requires
extracting $O(n/\sqrt{r})$ nodes from the persistent link-cut representation of $T_{i}$. If, instead of accessing
them one by one, we work with all of them at the same time, can be seen to take $O(n/\sqrt{r}\log r)$ time by the convexity of log.
\end{proof}

\end{document}